\documentclass[11pt]{article}

\usepackage{amsmath}
\usepackage{amssymb}
\usepackage{graphicx}
\usepackage{float}
\usepackage{caption}
\usepackage{amsthm}
\usepackage{color}
\definecolor{green}{rgb}{0,0.5977,0}
\usepackage[linesnumbered,algoruled,boxed,lined]{algorithm2e}

\newcommand{\rr}{\mathbb R}

\newcommand{\zz}{\mathbb Z}

\newcommand{\abs}[1]{\left|{#1}\right|}

\newcommand{\genseq}[3]{{#1}_1 {#3} {#1}_2 {#3} \dots {#3} {#1}_{#2}}
\newcommand{\seq}[2]{\genseq{#1}{#2}{,}}

\newcommand{\threecases}[6]{\begin{cases} #2 & #1 \\ #4 & #3 \\ #6 & #5 \end{cases}}

\newcommand{\txt}[1]{\text{#1}}

\newcommand{\stext}[1]{\ \ \ \ \ \text{(#1)}}

\newcommand{\snc}[1]{\stext{since ${#1}$}}
\newcommand{\indhyp}{\stext{by the inductive hypothesis}}

\newcommand{\push}{\\ & \ \ \ \ \ \ \ \ \ \ }

\newcommand{\ipnc}[3]{\begin{figure}[H]\begin{center}\includegraphics[scale = {#1}]{#2.pdf}\caption{#3}\end{center}\end{figure}}

\makeatletter 
\g@addto@macro{\@algocf@init}{\SetKwInOut{Parameter}{Parameters}} 
\makeatother

\theoremstyle{plain}
\newtheorem{theorem}{Theorem}
\newtheorem{lemma}[theorem]{Lemma}

\newtheorem{corollary}[theorem]{Corollary}

\theoremstyle{definition}

\numberwithin{theorem}{section}

\usepackage{url}

\newcommand{\sts}{\mathsf{sp}}
\newcommand{\emdash}{\,---\,}
\newcommand{\dist}{\mu^\star}

\usepackage{hyperref,xcolor}
\hypersetup{
    colorlinks = true,
 	linkcolor = teal,
 	citecolor = teal,
 	urlcolor = teal
}

\makeatletter
\algocf@newcmdside@kobe{WithProb@do}{%
	\KwSty{do}
	\ifArgumentEmpty{#1}\relax{ #1}%
	\algocf@block{#2}{end}{#3}%
	\par
}
\newcommand\WithProb[2]{%
	\KwSty{with probability} ${#1}$ \WithProb@do{#2}
}

\algocf@newcmdside@kobe{WithProbElse@do}{%
	\KwSty{do}
	\ifArgumentEmpty{#1}\relax{ #1}%
	\algocf@group{#2}
	\par
}
\algocf@newcmdside@kobe{WithProbElse@else}{%
	\KwSty{else}%
	\ifArgumentEmpty{#1}\relax{ #1}%
	\algocf@block{#2}{end}{#3}%
	\par
}
\newcommand\WithProbElse[3]{%
	\KwSty{with probability} ${#1}$ \WithProbElse@do{#2}
	\WithProbElse@else{#3}
}
\makeatother

\begin{document}
\title{Compact Redistricting Plans Have Many Spanning Trees}
\author{Ariel D. Procaccia\\
Harvard University
\and Jamie Tucker-Foltz\\
Harvard University}
	\maketitle
	\begin{abstract}
		In the design and analysis of political redistricting maps, it is often useful to be able to sample from the space of all partitions of the graph of census blocks into connected subgraphs of equal population. There are influential Markov chain Monte Carlo methods for doing so that are based on sampling and splitting random spanning trees. Empirical evidence suggests that the distributions such algorithms sample from place higher weight on more ``compact'' redistricting plans, which is a practically useful and desirable property. In this paper, we confirm these observations analytically, establishing an inverse \emph{exponential} relationship between the total length of the boundaries separating districts and the probability that such a map will be sampled. This result provides theoretical underpinnings for algorithms that are already making a significant real-world impact.
	\end{abstract}
	
\section{Introduction}\label{secIntro}

In April 2021, the US Census Bureau released the 2020 apportionment counts: the tally of residents of the 50 states, which determines the number of seats each state is entitled to in the House of Representatives for the next decade. This kicks off the complicated and contentious process of redrawing the states' congressional districts, which, we expect, will prove to be a historically impactful application of algorithms to societal questions.

Although this process of \emph{redistricting} is mandated by the US Constitution and further spelled out by the states, the law allows significant flexibility and\emdash even in small states\emdash an astronomical number of possible plans. Since the early 19th Century, partisan actors have exploited this flexibility to engineer plans that give their parties an unfair advantage, a phenomenon known as \emph{gerrymandering}.

From the viewpoint of computer science, ``good'' redistricting is a natural algorithmic problem, and indeed there is a longstanding interest in algorithms for redistricting~\cite{HOR72}. However, it is only in the last few years that computer scientists and mathematicians have joined the fight against gerrymandering in earnest~\cite{Duch18,Pro19b}. 

The algorithmic approach that has been most successful in terms of policy impact is that of sampling a distribution over all feasible plans by running a Markov chain Monte Carlo (MCMC) algorithm, thereby generating an ensemble of ``representative'' plans. Such ensembles have been used\emdash including in a number of successful legal challenges to state redistricting plans in Pennsylvania, North Carolina, Michigan, Wisconsin and Ohio\emdash to determine whether implemented plans are statistical outliers, which suggests that they resulted from gerrymandering~\cite{Duch18b,CR15,HRM17,CFP17}. This method has the advantage of being able to discern whether proposed plans are fair in light of each state's unique political geography. 

In our view, the most influential MCMC method is \emph{ReCom} (shorthand for ``Recombination'')~\cite{DDS21}. For the upcoming cycle, at least two redistricting commissions will rely on ReCom to evaluate redistricting plans~\cite{Chen21}: the Michigan Independent Citizens Redistricting Commission (which is vested with the authority to adopt redistricting plans for the state) and the Wisconsin People’s Maps Commission (which was appointed by the governor to prepare plans for consideration by the state legislature). 

ReCom starts from an arbitrary plan represented as a partition of a graph where the vertices are \emph{census blocks}, which form the building blocks of districts,\footnote{Some states use precincts or counties instead.} and there is an edge between two vertices if the corresponding census blocks are adjacent. In each step, ReCom randomly selects a pair of adjacent districts, merges them together, then re-partitions the merged region into two new districts. The algorithm accomplishes the re-partitioning step by uniformly sampling a spanning tree of the merged region from the set of all such spanning trees. It then attempts to cut an edge of the spanning tree so that the two subtrees induce two new districts with roughly equal populations (which is a constitutional requirement); if there are multiple such edges it selects uniformly at random among them, and if there are none it samples a new spanning tree. The Markov chain is run for a fixed number of steps and the final plan is returned.

The case for ReCom rests on its ability to generate plans consisting of \emph{compact} districts with regular shapes. A practical and well-studied measure for compactness in the graph partitioning setting is to count the total number of \emph{cut edges}\emdash the edges whose endpoints lie in different districts; plans with fewer cut edges are more compact \cite{DDS21, NPRP}. Empirically, ReCom does generate compact plans according to this measure. 

By contrast, a theoretical compactness result had been out of reach. It is known that, with some slight technical modifications to the recombination step, the stationary distribution of the ReCom chain is the \emph{spanning tree distribution}, where the probability of a plan is proportional to the number of forests that span its districts or, equivalently, the product over districts of the number of spanning trees of each district \cite{ReversibleReCom}. The compactness of ReCom-generated plans (assuming sufficient mixing\footnote{Empirically, ReCom mixes extremely quickly, though this has not be established in any formal sense. In fact, it is still an open question whether the state space is connected, even when the census block graph is a square grid! However, a breakthrough result by \cite{CFP17} shows that, for a reversible Markov chain like that of \cite{ReversibleReCom}, it is possible to conduct meaningful statistical outlier tests \emph{without} mixing.}), therefore, depends on the relation between the number of spanning trees and the number of cut edges of a plan. It is perhaps intuitive that such a relation exists; for example, a rectangular district with few cut edges has many spanning trees, whereas a snaky district that consists of the same number of census blocks has many cut edges and few spanning trees (see Figure \ref{figComparingTwoPlans}).

Our goal is to formalize this intuition and quantify it. We aim to provide theoretical underpinnings for the observed compactness of ReCom-generated plans, further justifying the important role of this algorithm in redistricting.

\subsection{Main Result and Technique}\label{subMainResult}

Our main result, Theorem~\ref{thmSpanningTreeDistribution}, is best understood through a corollary. To state it informally, consider two partitions $\mathcal{P}_1$ and $\mathcal{P}_2$ of the census block graph (which is planar) into $m$ districts of equal size, and let $\abs{\partial \mathcal{P}_1}$ and $\abs{\partial\mathcal{P}_2}$ denote their \emph{discrete perimeters} (total number of cut edges). In addition, denote the spanning tree distribution by $\dist$. 

\newpage

\noindent\textbf{Corollary~\ref{corMain}} (informal version).  
\emph{For any pair of $m$-partitions $\mathcal{P}_1$ and $\mathcal{P}_2$ of a planar graph $G$ such that the second-largest degrees of $G$ and its dual are upper-bounded by a constant,
    $$\frac{\Pr_{\dist}[\mathcal{P}_1]}{\Pr_{\dist}[\mathcal{P}_2]} \geq 2^{\Theta\left(\frac{\abs{\partial\mathcal{P}_2}}{\abs{\partial\mathcal{P}_1}}\right)}.$$}%

In words, the corollary establishes an asymptotic \emph{exponential} relationship between the ratio of probabilities under $\dist$ of the two partitions and the inverse ratio of their discrete perimeters. It is reassuring that the relationship is exponential (rather than, say, linear): As noted by DeFord et al.~\cite{DDS21}, the space of all redistricting plans is generally dominated by non-compact plans, so a significant skew towards compact plans is required for it to be likely that such a plan would be sampled. We emphasize that this result does not imply anything about the probability of say, sampling a partition with less than a given number of cut edges. For that, one would additionally need bounds on the relative numbers of balanced partitions of a given compactness that exist.

Theorem~\ref{thmSpanningTreeDistribution} itself is not asymptotic; rather, it gives a precise relationship between the probabilities under $\dist$ and the discrete perimeters of different partitions. This relationship, in turn, depends on degree bounds that we formalize in Section~\ref{subBoundedness}. For example, for large grid graphs the theorem implies that if $|\partial\mathcal{P}_2|\geq 7.23\times |\partial\mathcal{P}_1|$ then $\Pr_{\dist}[\mathcal{P}_1]\geq \Pr_{\dist}[\mathcal{P}_2]$; the bounds are quite similar for planar graphs corresponding to real redistricting instances. 

The basic idea behind the proof of the theorem is to take an arbitrary partition $\mathcal{P}$ of $G$ into $m$ districts and add $m - 1$ edges to connect the subgraphs of each district together into a single connected graph $H$ (see Figure \ref{figAddTreeEdges}). It is not too hard to see that the number of cut edges of $\mathcal{P}$ is precisely $\abs{E(G)} - \abs{E(H)} + (m - 1)$, and the likelihood of sampling $\mathcal{P}$ is proportional to the number of spanning trees of $H$. Thus, it suffices to establish a relationship between the number of edges in $H$ and the number of spanning trees in $H$.

Our approach is to imagine iteratively removing edges from $G$ until we are left with $H$ and compute upper and lower bounds for the average factor by which the number of spanning trees decreases at each iteration (Lemma \ref{lemCountingTreesBoundsGeneral}). Specifically, we are interested in bounding the probability that a given edge is contained in a uniformly random spanning tree of the current graph. This quantity is known as the \emph{effective resistance} of the edge due to an alternative, equivalent definition in terms of the electrical resistance across the edge in a network of resistors (see Section \ref{subEffectiveResistance}). We use the electrical formulation of the problem to derive useful bounds for counting spanning trees in our particular setting. While these bounds do not imply that the effective resistance of each edge deleted from $G$ is upper and lower bounded by fixed constants at every iteration (as this is generally not true), we argue that the geometric mean of the effective resistances is. For this, we use a potential function to amortize the extremely high and low factors over the less extreme factors accumulated over previous iterations.

\subsection{Related Work}\label{subLitReview}

The ReCom Markov chain is one of many Markov chains on the space of graph partitions that have been studied in the context of redistricting \cite{DDS21, Chain1, Chain2, CFP17, Forests}. Most of the predecessors of ReCom are based on the ``Flip'' chain, whereby a single census block on the boundary of a district is reassigned at each step. While computing the transition function is easier for Flip than for ReCom, random walks using the Flip chain mix extremely slowly, and produce non-compact districts by any reasonable metric. Empirically speaking, ReCom is a great improvement over Flip.

From a theoretical perspective, very little is known about the properties of the spanning tree distribution from which ReCom samples. DeFord et al.~\cite{DDS21} give some informal, intuitive arguments for why we should expect it to favor compact partitions\emdash for example, it is easy to see that adjoining a long ``tentacle'' to one of the districts in an otherwise compact partition will reduce the number of spanning trees by a large factor. It is conjectured that, among all grid subgraphs of the same number of vertices, square subgrids (which have minimal perimeter) have the largest number of spanning trees. Kenyon \cite{PriorTheory1} provides an encouraging result to this end, asymptotically counting spanning trees in finer and finer grids approximating a rectilinear polygon in $\rr^2$. A very recent paper by Tapp \cite{Tapp} gives concrete, non-asymptotic bounds on the number of spanning trees of a grid subgraph in terms of its perimeter and number of vertices, but they are not strong enough to resolve this conjecture either.

There is a large body of work in the combinatorics literature on approximately counting spanning trees in graphs of bounded degree (see, for example, \cite{TreesAlon, TreesUpperAndLower, TreesMcKay, TreesInCirculent}). Given two subgraphs $H_1$ and $H_2$ of $G$ as in Section \ref{subMainResult}, a natural line of attack for our main result is to apply these bounds to each $H_i$. However, these bounds are insufficient for our purposes since they have a multiplicative error term which is exponential in the number of vertices of $H_i$. For example, if we assume $G$ is 3-regular, it is known that the number of spanning trees in $H_i$ is bounded between $1.62^{\abs{V(H_i)}}$ and $2.31^{\abs{V(H_i)}}$ \cite{TreesUpperAndLower, TreesMcKay}. Our result requires the error to be on the order of $2^{O(\abs{E(G)} - \abs{E(H_i)})}$, which can be significantly smaller than $2^{\Theta(\abs{V(H_i)})}$. Our alternative analysis based on effective resistances overcomes this difficulty, since it does not accumulate error for each edge in $E(H_i)$, but instead for the edges in $E(G) \setminus E(H_i)$.

Connections between effective resistance and discrete perimeters have been studied before. A well-known example is the Nash-Williams Inequality \cite[(2.13)]{NashWilliams}, which yields a lower bound for effective resistance in terms of sets of edges that separate the graph. Benjamini and Kozma \cite{BenjaminiKozma} give an upper bound for effective resistance via sums of isoperimetric quantities for connected sets containing the two vertices. These results are largely orthogonal to our work, with the exception of Footnote~\ref{ftn}.

Our contribution can be viewed as a positive result about the computational tractability of approximately sampling graph partitions from ``nice'' distributions. By contrast, Najt, Deford, and Solomon \cite{NPRP} establish hardness of several related sampling problems motivated by redistricting, mostly via reductions from the Hamiltonian Cycle problem. For example, they show that, for any $\lambda \in (0, 1]$, there is no polynomial time algorithm to approximately sample $k$-partitions of an input graph $G$ proportional to $\lambda^{\abs{\txt{cut edges}}}$ unless $\txt{NP} = \txt{RP}$. Even if ReCom could be shown to run in polynomial time, this result still would not contradict ours because (1) the graphs produced by the reduction do not satisfy the (realistic) conditions of our main theorem, (2) the approximation guarantee is much more stringent than ours, and (3) their result allows for partitions that are not even approximately balanced.

\section{Preliminaries}\label{secPreliminaries}

All graphs we consider are undirected and unweighted, but may have multiple edges and/or self-loops. For any connected, planar embedded graph $G$, we write $G^*$ for the dual of $G$, which is the graph whose vertex set consists of the faces of $G$ with respect to the embedding, where there is an edge in $G^*$ between two faces whenever they share a common edge on their boundaries in $G$. For any face $f \in V(G^*)$, we overload the notation $\deg(f)$ to mean the degree of $f$ as a vertex in $G^*$, i.e., the number of edges/vertices on the boundary of $f$ in $G$.

\subsection{Graph Partitions}\label{subSpanningTreeDistributionDescription}

Our central object of study is the \emph{census block graph}, in which the vertices represent census blocks, and there is an edge between two vertices if the census blocks share a border of nonzero length (or are legally considered adjacent for other reasons, e.g., in the case of islands). See Figure \ref{figComparingTwoPlans} for an example of a census block graph.

\ipnc{.7}{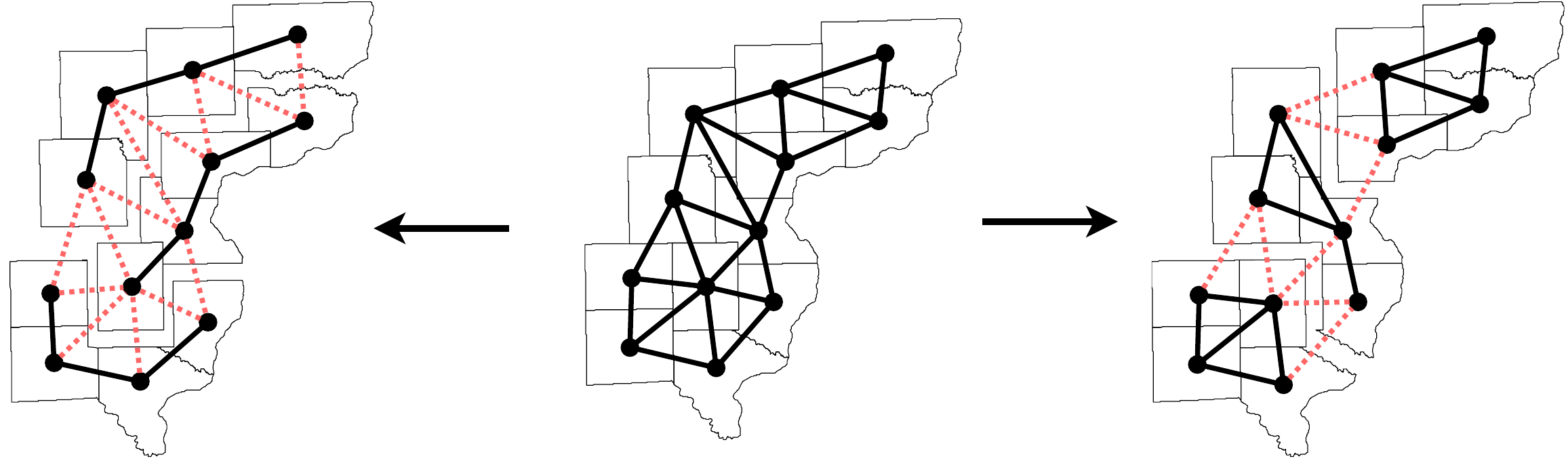}{\label{figComparingTwoPlans}In center, 12 counties in the Southeast corner of Iowa and the corresponding census block graph (in Iowa, counties are actually the atomic units for redistricting, not census blocks). On the left and right, two potential 3-partitions into connected districts of 4 counties each, where the red dashed lines are cut edges.}

For simplicity, we consider an idealized redistricting setting where all census blocks have equal population and districts must be exactly population-balanced. Thus, we define an \emph{$m$-partition} of a graph $G$ to be a partition $\mathcal{P} = \{\seq{D}{m}\}$ of the vertex set of $G$ such that each $D_i$, which we call a \emph{district}, has exactly $\frac{\abs{V(G)}}{m}$ vertices and induces a connected subgraph. Given an $m$-partition $\mathcal{P}$ of $G$, we write $G / \mathcal{P}$ to denote the graph obtained by contracting the induced subgraphs of all the districts. A \emph{cut edge} of an $m$-partition $\mathcal{P}$ is an edge with endpoints in different districts. For example, the 3-partition on the left of Figure \ref{figComparingTwoPlans} has 13 cut edges, while the 3-partition on the right only has 8. We write $\partial \mathcal{P}$ for the set of cut edges of $\mathcal{P}$. 

The \emph{spanning tree score} of a graph $G$, written $\sts(G)$, is the number of spanning trees of $G$. The \emph{spanning tree score} of an $m$-partition $\mathcal{P}$, written $\sts(\mathcal{P})$, is defined as the product of the spanning tree scores of the induced subgraphs of each of the districts in $\mathcal{P}$. The spanning tree score of the 3-partition on the left of Figure \ref{figComparingTwoPlans} is $1 \times 1 \times 1 = 1$, whereas the spanning tree score of the 3-partition on the right is $8 \times 3 \times 8 = 192$.

\subsection{Planar Graphs With Bounded Vertex and Face Degrees}\label{subBoundedness}

Typical instances to the graph partitioning problem that arise in redistricting have several additional properties:
\begin{itemize}
	\item The census block graph is connected, planar, and does not contain any self-loops, leaves, or bridges.\footnote{It does sometimes occur that one precinct will be surrounded by another one on all sides, in which case it is a leaf. However, this is rare, and can easily be modeled by just merging the two precincts together.}
	\item All census blocks have low degree.
	\item No large group of census blocks intersect at the same boundary point (e.g., in the graph of states in the USA, there is a ``Four Corners'' location between Colorado, Utah, Arizona, and New Mexico, but there is no ``Five Corners'' or greater).
\end{itemize}
This motivates the following definition. For any positive integers $k_1$ and $k_2$, we say that $G$ is \emph{$(k_1, k_2)$-bounded} if $G$ is connected, neither $G$ nor $G^*$ have a self loop, and there exists a planar embedding of $G$, a vertex $v_0 \in V(G)$, and a face $f_0 \in V(G^*)$ such that, for all $v \in V(G) \setminus \{v_0\}$, $\deg(v) \leq k_1$, and for all $f \in V(G^*) \setminus \{f_0\}$, $\deg(f) \leq k_2$.

Simply put, $k_1$ and $k_2$ upper bound the second-largest degrees in $G$ and $G^*$, respectively. Think of $f_0$, the face of unbounded degree, as the outer face of the census block graph. It would be impractical to impose a bound on the number of census blocks this face touches. The vertex $v_0$ has no specific meaning in our redistricting context; we allow for such a vertex merely for generality and symmetry.

For example, grid graphs are $(4, 4)$-bounded, the subgraph of counties in Iowa shown in Figure \ref{figComparingTwoPlans} is $(5, 3)$-bounded, and the entire graph of all Iowa counties happens to be $(6, 4)$-bounded.

\subsection{Effective Resistance}\label{subEffectiveResistance}

We now briefly review some tools from spectral graph theory that we will need shortly. For more background, we refer the reader to Chapters 12 and 13 of Spielman \cite{SAGT}.

Consider the following physics problem. We are given a graph $G$ and a specific edge $e^* \in E(G)$. We place a resistor of unit resistance on every edge (including $e^*$), hook up a power supply between the endpoints of $e^*$ (call them $a$ and $b$), and adjust the voltage so that 1 unit of current is flowing into $a$ and out of $b$. The \emph{effective resistance} of $e^*$, denoted $R_{ab}$, is defined as the voltage differential between $a$ and $b$ under this setup. 

Formally, this voltage difference can be computed by enforcing Ohm's law ``$V = IR$'' (voltage equals current times resistance) across every edge. Specifically, we wish to find voltages $v(c)$ for every vertex $c$ and currents $i(e)$ for every oriented edge $e$ such that:
\begin{itemize}
	\item There is 1 net flow out of $a$.
	\item There is 1 net flow into $b$.
	\item For all $c \notin \{a, b\}$ there is zero net flow in/out of $c$.
	\item (Without loss of generality) $v(b) = 0$.
	\item For any edge $e$ from vertex $c$ to vertex $d$, $v(c) - v(d) = i(e)$.
\end{itemize}
Given any edge $e^*$, and picking the arbitrary orientation of $e^*$ from $a$ to $b$, there is a unique solution of $v(\cdot)$ and $i(\cdot)$ satisfying these constraints. The effective resistance of $e^*$ is $R_{ab} = v(a) = i(e^*)$.

\ipnc{.8}{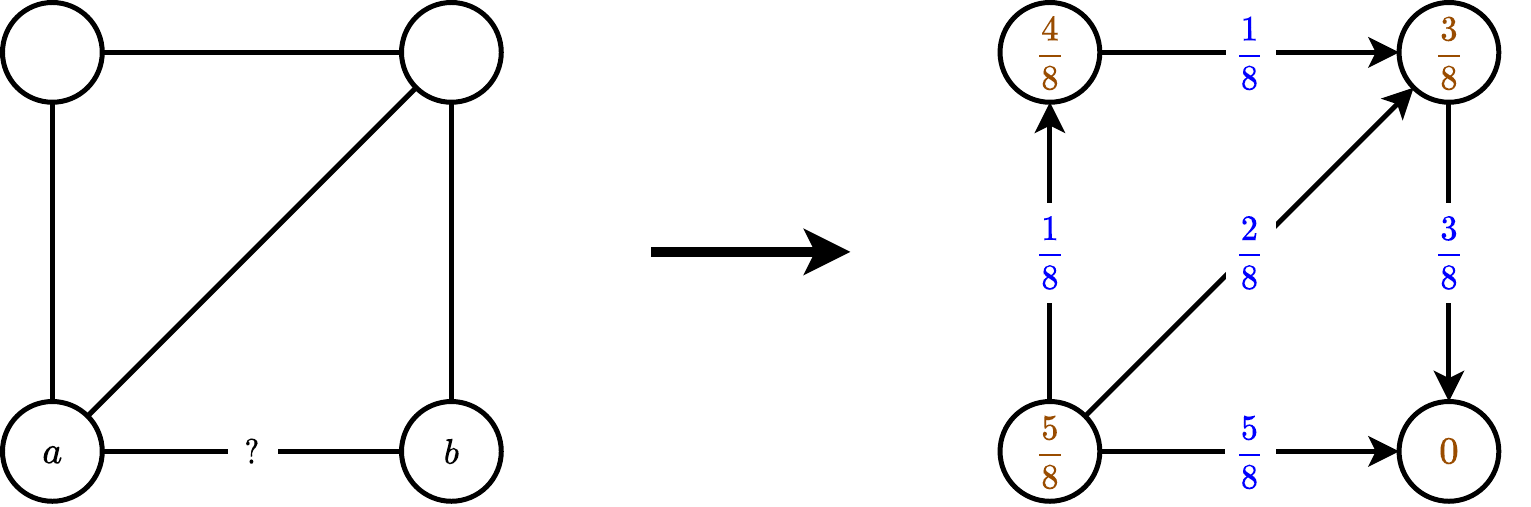}{\label{figEffectiveResistanceExample}Computing the effective resistance across the edge between $a$ and $b$.}

For example, consider the graph on the left in Figure \ref{figEffectiveResistanceExample}. To determine the effective resistance between $a$ and $b$, we compute the unique voltages (brown) and currents (blue) satisfying the constraints, as shown on the right. The voltage at $a$ and the current from $a$ to $b$ are both $\frac58$, so the effective resistance is $R_{ab} = \frac58$. 

Our interest in effective resistance stems from the following well-known statement, which gives an equivalent definition in terms of spanning trees.

\begin{lemma}\label{lemEffectiveResistanceAndTrees}
	For any edge $e^* \in E(G)$ between two vertices $a, b \in V(G)$, the effective resistance $R_{ab}$ is equal to the probability that $e^*$ is in a uniformly chosen spanning tree of $G$.
\end{lemma}

For example, one can verify by enumeration that the graph from Figure~\ref{figEffectiveResistanceExample} has 8 spanning trees, and exactly 5 of them include the edge $\{a, b\}$.

To derive useful bounds on effective resistance (and, therefore, on spanning trees) we briefly consider a more general version of the problem in which there are resistors in the graph with non-unit resistance. In that case, we simply replace the final condition with
$$v(c) - v(d) = i(e)r(e),$$
where $r(e)$ is the resistance of $e$. It is not too hard to see that deleting an edge is equivalent to setting its resistance to $\infty$, so that no current can possibly flow through it, while contracting an edge is equivalent to setting its resistance to $0$, so that both endpoints must have the same voltage.

\begin{lemma}[Rayleigh’s Monotonicity Principle]\label{lemRayleigh}
	For any vertices $a$ and $b$ of a graph $G$, weakly increasing the resistance of the resistor on any edge in $G$ weakly increases $R_{ab}$.
\end{lemma}

The following lemmas use Rayleigh’s Monotonicity Principle to derive upper and lower bounds for effective resistances.\footnote{\label{ftn}We note that Lemma \ref{lemERDegree} can alternatively be derived as a special case of the Nash-Williams Inequality \cite[(2.13)]{NashWilliams} with one cutset.}

\begin{lemma}\label{lemERCycle}
	Let $a, c_1, c_2, \dots, c_{k - 2}, b$ be a simple cycle of length $k \geq 2$ in a network of unit resistors. Then $R_{ab} \leq 1 - \frac{1}{k}$.
\end{lemma}

\begin{proof}
	Delete all edges (i.e., send resistances to infinity) except for the cycle. A simple calculation shows that, in the new network, $R_{ab} = 1 - \frac{1}{k}$, with $\frac{1}{k}$ units of current passing the ``long way'' around the cycle and $1 - \frac{1}{k}$ units of current passing through the given edge from $a$ to $b$. By Lemma \ref{lemRayleigh}, the effective resistance in the original network must be at most $1 - \frac{1}{k}$.
\end{proof}

\begin{lemma}\label{lemERDegree}
	Let $a$ and $b$ adjacent vertices in a network of unit resistors. Then $R_{ab} \geq \frac{1}{\deg(a)}$.
\end{lemma}

\begin{proof}
	Contract all edges (i.e., lower resistances to zero) except for the edges from $a$. Then add additional resistors of resistance zero joining each neighbor of $a$ to $b$ if not already adjacent (i.e., lower resistances from infinity to zero). A simple calculation shows that, in the new network, $R_{ab} = \frac{1}{\deg(a)}$, with $\frac{1}{\deg(a)}$ units of current on every edge from $a$, and zero current on every other edge. By Lemma \ref{lemRayleigh}, the effective resistance in the original network must be at least $\frac{1}{\deg(a)}$.
\end{proof}

\section{Main Result}\label{secProof}

In this section, we prove the following theorem, which shows that, for sufficiently large census block graphs, the spanning tree distribution\emdash denoted hereinafter by $\dist$\emdash favors partitions with smaller boundaries.

\begin{theorem}\label{thmSpanningTreeDistribution}
	For any positive integers $k_1$ and $k_2$, any $\alpha \geq 1$, and any $\varepsilon > 0$, let
	\begin{equation}\label{equLambdaDefinition}
	    \lambda = \lambda(k_1, k_2, \alpha, \varepsilon) := \frac{\log\left(\frac{1}{2k_2}\right) - \log(\alpha)}{\log\left(1 - \frac{1}{k_1}\right)} + \varepsilon.
	\end{equation}
	For any two $m$-partitions $\mathcal{P}_1$ and $\mathcal{P}_2$ of a $(k_1, k_2)$-bounded graph, if
	\begin{equation}\label{equPremiseBoundaryGap}
		\abs{\partial\mathcal{P}_2} \geq \lambda \abs{\partial\mathcal{P}_1}
	\end{equation}
	and
	\begin{equation}\label{equPremiseEpsilon}
		\abs{\partial\mathcal{P}_1} \geq \frac{m - 1}{\varepsilon},
	\end{equation}
	then $\Pr_{\dist}[\mathcal{P}_1]\geq \alpha\Pr_{\dist}[\mathcal{P}_2]$.
\end{theorem}

For example, in grid graphs, where $k_1 = k_2 = 4$, we have
$$\lambda(4, 4, 1, \varepsilon) = \frac{\log\left(\frac{1}{8}\right) - \log(1)}{\log\left(1 - \frac{1}{4}\right)} + \varepsilon \approx 7.23$$
(for small $\varepsilon$). Thus, our result shows that, for fixed $m$ and sufficiently large grids, an $m$-partition whose total boundary length is at least 7.23 times longer than that of another partition is less likely to be sampled.

To prove Theorem \ref{thmSpanningTreeDistribution}, we first establish upper and lower bounds for the geometric mean of the $p_i$ values on any run of the following sampling algorithm. \vspace{.1cm}

\begin{algorithm}[H]
	\caption{\label{algTreeSample}Samples uniformly from the set of all spanning trees of an input graph $G$.}
	$i \gets 0$\;
	$T \gets \{\}$\;
	\While{$G$ has at least 2 vertices}
	{
		$i \gets i + 1$\;
		$e_i \gets$ arbitrary edge in $G$\;\label{linChooseArbitraryEdge}
		$r_i \gets$ effective resistance of $e_i$\;
		\WithProbElse{r_i}
		{
			$p_i \gets r_i$\;\label{linPi1}
			$T \gets T \cup \{e_i\}$\;
			contract $e_i$ in $G$ (keeping multiple edges and self-loops)\;
		}
		{
			$p_i \gets 1 - r_i$\;\label{linPi2}
			delete $e_i$ from $G$\;
		}
	}
	\KwRet{$T$}\;
\end{algorithm}
\vspace{.15cm}

By Lemma \ref{lemEffectiveResistanceAndTrees}, at any point in the execution of Algorithm \ref{algTreeSample}, the probability of the computation path (which is the product of all $p_i$ values so far) is equal to the probability that a randomly chosen spanning tree of the original input graph $G$ includes all of the contracted edges and does not include any of the deleted edges. In particular, this implies that the order in which the edges are chosen in the successive executions of line \ref{linChooseArbitraryEdge} does not affect the probability of the given computation path. By the time the algorithm terminates, $T$ is guaranteed to be a uniform sample from the set of all spanning trees.

Our main technical lemma consists of two statements, where statement (\ref{itmLemmaGeneral}) is more general but statement (\ref{itmLemmaSpecific}) gives tighter bounds. While we only require (\ref{itmLemmaSpecific}) for our application, we additionally prove (\ref{itmLemmaGeneral}) because we believe it may be of independent interest.

\begin{lemma}\label{lemCountingTreesBoundsGeneral}
	For any positive integers $k_1$ and $k_2$, there exist constants $0 < c_1 < c_2 < 1$, such that, on any run of Algorithm \ref{algTreeSample} on a $(k_1, k_2)$-bounded graph, after any number of iterations $t$,
	$${c_1}^t \leq p_1 p_2 p_3 \dots p_t \leq {c_2}^{t}$$
	(where the $p_i$ are as defined on Lines \ref{linPi1} and \ref{linPi2}). Specifically, this holds with the following constants:
	\begin{enumerate}
		\item\label{itmLemmaGeneral} If the run involves both deletions and contractions,
		\begin{align*}
			c_1 &= \frac{1}{2\max\{k_1, k_2\}}, & c_2 &= \left(1 - \frac{1}{\max\{k_1, k_2\}}\right)^{\frac{1}{2(\min\{k_1, k_2\} - 1)}}.
		\end{align*}
		\item\label{itmLemmaSpecific} If the run involves only deletions,
		\begin{align*}
			c_1 &= \frac{1}{2k_2}, & c_2 &= 1 - \frac{1}{k_1}.
		\end{align*}
	\end{enumerate}
\end{lemma}

We remark that it is \emph{not} true that each $p_i$ is always between $c_1$ and $c_2$. It is not too hard to see that sometimes we may have $r_i = 0$ or $r_i = 1$, in which case $p_i = 1$. Also, after an adversarial sequence of contractions, it is possible to have $r_i$ arbitrarily close to 0 but not equal to 0, and after an adversarial sequence of deletions, it is possible to have $r_i$ arbitrarily close to 1 but not equal to 1. Thus, $p_i$ may be arbitrarily close to 0 as well. In these scenarios, however, it can take many iterations to get to such a case, so we must amortize these bad factors over the iterations where $p_i$ is less extreme.

\begin{proof}[Proof of Lemma~\ref{lemCountingTreesBoundsGeneral}]
	Let $v_0$ and $f_0$ be as in the definition of $G$ being $(k_1, k_2)$-bounded. We begin by proving the upper bounds. Let $D$ be the subset of the first $t$ edges that are ultimately deleted, and let $C$ be the subset of the first $t$ edges that are ultimately contracted. For every edge $e \in C$, choose a face $f(e) \in V(G^*)$ such that $e$ is on the boundary of $f(e)$ and $\deg(f(e)) \leq k_2$. Note that this is always possible since the two faces $e$ bounds cannot both be $f_0$, for this would imply $G^*$ has a self-loop, contradicting the definition of $G$ being $(k_1, k_2)$-bounded. Partition $C$ into $C_1 \cup C_2$, where $C_1$ consists of the edges $e$ such that $f(e)$ does not contain any edges in $D$. Without loss of generality, we assume Algorithm \ref{algTreeSample} first processes the edges in $D$, then in $C_1$, then in $C_2$.
	
	By Lemma \ref{lemERDegree}, each edge in $D$ has effective resistance at least $\frac{1}{k_1}$ when it is deleted from $G$. Therefore, on a deletion iteration $i$, we have
	$$p_i = 1 - r_i \leq 1 - \frac{1}{k_1}.$$
	This immediately implies the upper bound in statement (\ref{itmLemmaSpecific}).
	
	To prove the upper bound in statement (\ref{itmLemmaGeneral}), we must consider the contractions as well. By Lemma \ref{lemERCycle} and the way we chose $C_1$, each edge in $C_1$ has effective resistance at most $1 - \frac{1}{k_2}$ when contracted in $G$, so on an iteration $i$ that contracts an edge from $C_1$,
	$$p_i = r_i \leq 1 - \frac{1}{k_2}.$$
	
	We next claim that $\abs{D} + \abs{C_1} \geq \frac{t}{2k_2 - 1}$. Supposing for contradiction that this were not the case, we must have $\abs{D} < \frac{t}{2k_2 - 1}$ and
	$$\abs{C_2} = t - (\abs{D} + \abs{C_1}) > t - \frac{t}{2k_1 - 1} = \frac{2(k_2 - 1)t}{2k_2 - 1},$$
	so it follows that
	$$\abs{C_2} > 2(k_2 - 1)\abs{D}.$$
	This contradicts the way $C_2$ was defined, since each edge in $D$ can be contained in $f(e)$ for at most $2(k_2 - 1)$ edges $e \in C_2$ (in the extreme case, the edge in $D$ lies between two faces, each containing $k_2 - 1$ other edges $e$).
	
	Putting these bounds together, we have
	\begin{align*}
		\prod_{1 \leq i \leq t} p_i &= \left(\prod_{1 \leq i \leq t,\ e_i \in D} p_i\right) \left(\prod_{1 \leq i \leq t,\ e_i \in C_1} p_i\right) \left(\prod_{1 \leq i \leq t,\ e_i \in C_2} p_i\right)\\
		&\leq \left(1 - \frac{1}{k_1}\right)^{\abs{D}} \left(1 - \frac{1}{k_2}\right)^{\abs{C_1}} (1)\\
		&\leq \left(1 - \frac{1}{\max\{k_1, k_2\}}\right)^{\abs{D} + \abs{C_1}}\\
		&\leq \left(1 - \frac{1}{\max\{k_1, k_2\}}\right)^{\frac{t}{2(k_2 - 1)}}\\
		&= \left(\left(1 - \frac{1}{\max\{k_1, k_2\}}\right)^{\frac{1}{2(k_2 - 1)}}\right)^t.
	\end{align*}
	Note that we could make the dual argument, first processing the contractions, in which case we would be left with the same upper bound, except with a $k_1$ in the exponent instead of $k_2$. The upper bound in statement (\ref{itmLemmaGeneral}) follows.
	
	To prove the lower bounds, we define a potential function on the graph as follows. Initially, place one pebble on every vertex and face of $G$ except $v_0$ and $f_0$, which receive piles of $\deg(v_0)$ and $\deg(f_0)$ pebbles, respectively. Throughout the execution of Algorithm \ref{algTreeSample}, whenever an edge is deleted from $G$ (and its dual edge is contracted in $G^*$), combine the piles on the faces on either side of the deleted edge into a new pile on the new face, and whenever an edge is contracted in $G$ (and its dual edge is deleted from $G^*$), combine the piles on the endpoints of the contracted edge into a new pile on the new vertex. After each iteration $0 \leq i \leq t$, let $P_i$ denote the product of the numbers of pebbles in each pile. When $i = 0$, before any edges have been deleted or contracted, $P_0 = \deg(v_0) \deg(f_0)$.
	
	We claim that, after any deletion iteration $1 \leq i \leq t$,
	\begin{equation}\label{equPotentialConstantDeletion}
		p_i \frac{P_{i - 1}}{P_i} \geq \frac{1}{2k_2},
	\end{equation}
	and after any contraction iteration $1 \leq i \leq t$,
	\begin{equation}\label{equPotentialConstantContraction}
		p_i \frac{P_{i - 1}}{P_i} \geq \frac{1}{2k_1}.
	\end{equation}
	The proofs of these two statements are completely dual, so we will only discuss the deletion case.
	
	Suppose that the deleted edge on round $i$ lies between faces $f_1$ and $f_2$. Suppose there are $x$ pebbles on $f_1$ and $y$ pebbles on $f_2$, and, without loss of generality, assume $x \leq y$. Let $z$ be the product of the number of pebbles in all of the other piles before the edge is deleted from $G$. By Lemma \ref{lemERCycle}, we know $p_i = 1 - r_i \geq \frac{1}{\deg(f_1)}$. Observe that, before any edges are deleted, the degree of every face is at most $k_2$ times the number of pebbles on that face. This is because, initially, either the degree is at most $k_2$ or, in the case of $f_0$, the number of pebbles is equal to the degree. It is not too hard to see that this property is preserved under contracting edges (which can only lower degrees) and deleting edges (which simultaneously combines face degrees and pebble pile sizes). Thus, this property holds of $f_1$ before deletion on round $i$, i.e., $\deg(f_1) \leq k_2x$. Therefore,
	\begin{align*}
		p_i \frac{P_{i - 1}}{P_i} &\geq \frac{1}{\deg(f_1)} \frac{P_{i - 1}}{P_i}\\
		&\geq \frac{1}{k_2x} \cdot \frac{P_{i - 1}}{P_i}\\
		&= \frac{1}{k_2x} \cdot \frac{xyz}{(x + y)z}\\
		&\geq \frac{1}{k_2x} \cdot \frac{xyz}{(2y)z} \snc{x \leq y}\\
		&= \frac{1}{2k_2},
	\end{align*}
	as desired.
	
	If $I_D$ is the set of deletion iterations and $I_C$ is the set of contraction iterations, it follows from Equations (\ref{equPotentialConstantDeletion}) and (\ref{equPotentialConstantContraction}) that
	\begin{align*}
		p_1 p_2 \dots p_{t} &\geq \prod_{i \in I_D} \left(\frac{1}{2k_2} \cdot \frac{P_i}{P_{i - 1}}\right) \prod_{i \in I_C} \left(\frac{1}{2k_1} \cdot \frac{P_i}{P_{i - 1}}\right)\\
		&= \frac{P_1}{P_{0}} \cdot \frac{P_2}{P_{1}} \cdot \dots \cdot \frac{P_t}{P_{t - 1}} \prod_{i \in I_D} \left(\frac{1}{2k_2}\right) \prod_{i \in I_C} \left(\frac{1}{2k_1}\right)\\
		&= \frac{P_{t}}{P_0} \left(\frac{1}{2k_2}\right)^{\abs{I_D}} \left(\frac{1}{2k_1}\right)^{\abs{I_C}}\\
		&\geq \left(\frac{1}{2k_2}\right)^{\abs{I_D}} \left(\frac{1}{2k_1}\right)^{\abs{I_C}},
	\end{align*}
	where the final inequality holds since the initial piles of pebbles placed on $v_0$ and $f_0$ can only grow and cannot merge, so $P_t \geq \deg(v_0) \deg(f_0) = P_0$. The lower bound for statement (\ref{itmLemmaGeneral}) follows since $\abs{I_D} + \abs{I_C} = t$, while the lower bound for statement (\ref{itmLemmaSpecific}) follows by specializing $\abs{I_D} = t$ and $\abs{I_C} = 0$.
\end{proof}

We are now ready to prove our main result and its corollary.

\begin{proof}[Proof of Theorem \ref{thmSpanningTreeDistribution}]
	Letting $c_1$ and $c_2$ be as in Lemma \ref{lemCountingTreesBoundsGeneral} (\ref{itmLemmaSpecific}), observe that
	$$\lambda = \frac{\log(c_1) - \log(\alpha)}{\log(c_2)} + \varepsilon = \frac{\log(c_1)}{\log(c_2)} + \varepsilon - \log_{c_2}(\alpha) \geq \frac{\log(c_1)}{\log(c_2)} + \varepsilon - \frac{\log_{c_2}(\alpha)}{\abs{\partial \mathcal{P}_1}}.$$
	Rearranging, we have
	$$\lambda - \varepsilon + \frac{\log_{c_2}(\alpha)}{\abs{\partial \mathcal{P}_1}}\geq \frac{\log(c_1)}{\log(c_2)},$$
	so
	\begin{equation}\label{equLambdaProperty}
	    \alpha c_2^{\abs{\partial \mathcal{P}_1}(\lambda - \varepsilon)} \leq c_2^{\abs{\partial \mathcal{P}_1} \cdot (\log(c_1)/\log(c_2))}.
	\end{equation}
	
	For each $i \in \{1, 2\}$, let $T_i$ be a spanning tree of $G / \mathcal{P}_i$, and let $S_i \subseteq \partial \mathcal{P}_i$ be a set of edges of size $\abs{S} = \abs{\partial \mathcal{P}_i} - (m - 1)$ obtained by removing from $\partial \mathcal{P}_i$ one cut edge between districts $D, D' \in \mathcal{P}_i$ for every pair of adjacent districts $\{D, D'\} \in E(T_i)$, as shown in Figure \ref{figAddTreeEdges}. The probability of the partial computation path of Algorithm \ref{algTreeSample} which deletes all edges in $S_i$ is the probability of drawing from the uniform distribution a tree with no edges in $S_i$. By construction, the number of such trees is equal to the spanning tree score of $\mathcal{P}_i$. Therefore, applying Lemma \ref{lemCountingTreesBoundsGeneral}, there exist $0 < c_1 < c_2 < 1$ such that, for each $i \in \{1, 2\}$,
	\begin{equation}\label{equSpanningTreeScoreUpperAndLowerBounds}
		c_1^{\abs{\partial \mathcal{P}_i} - m + 1} \leq \frac{\sts(\mathcal{P}_i)}{\sts(G)} \leq c_2^{\abs{\partial \mathcal{P}_i} - m + 1}.
	\end{equation}
	
	\ipnc{.7}{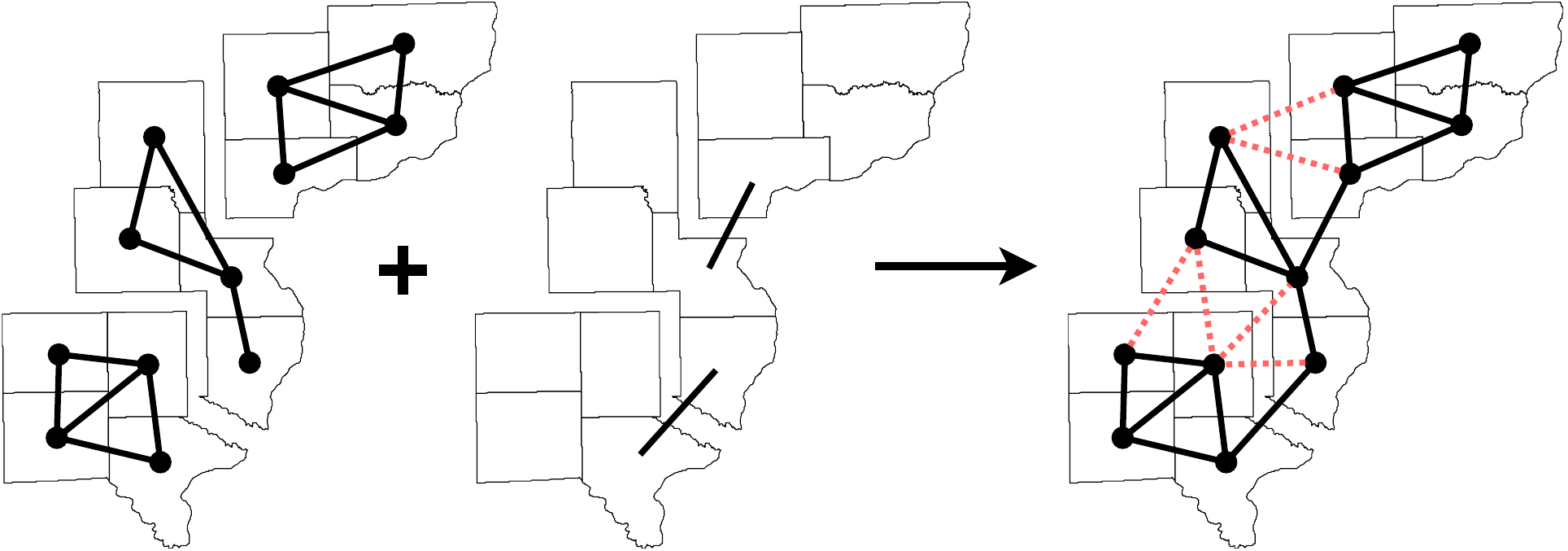}{\label{figAddTreeEdges}Illustration of the proof of Theorem \ref{thmSpanningTreeDistribution}, where we approximately compute the spanning tree score of the partition on the left by linking the districts together with additional edges. The $S_i$ set consists of all the dashed red edges in the graph on the right. After removing these edges, the number of spanning trees of the graph on the right is precisely the spanning tree score of the original partition (which in this case is 192).}
	
	Let $\beta$ be the normalization constant such that, for any $m$-partition $\mathcal{P}$ of $G$, $\Pr_{\dist}[\mathcal{P}] = \frac{\beta}{\sts(G)}\sts(\mathcal{P})$. Then, whenever $\mathcal{P}_1$ and $\mathcal{P}_2$ satisfy Equations (\ref{equPremiseBoundaryGap}) and (\ref{equPremiseEpsilon}), we derive that
	\begin{align*}
		\Pr_{\dist}[\mathcal{P}_1] &= \beta\frac{\sts(\mathcal{P}_1)}{\sts(G)}\\
		&\geq \beta c_1^{\abs{\partial \mathcal{P}_1} - m + 1} \stext{from Equation (\ref{equSpanningTreeScoreUpperAndLowerBounds})}\\
		&\geq \beta c_1^{\abs{\partial \mathcal{P}_1}}\\
		&= \beta c_2^{\abs{\partial \mathcal{P}_1} \cdot (\log(c_1)/\log(c_2))}\\
		&\geq \alpha\beta c_2^{\lambda\abs{\partial \mathcal{P}_1} - \varepsilon \abs{\partial \mathcal{P}_1}} \stext{from Equation (\ref{equLambdaProperty})}\\
		&\geq \alpha\beta c_2^{\abs{\partial \mathcal{P}_2} - \varepsilon \abs{\partial \mathcal{P}_1}} \stext{from Equation (\ref{equPremiseBoundaryGap})}\\
		&\geq \alpha\beta c_2^{\abs{\partial \mathcal{P}_2} - m + 1} \stext{from Equation (\ref{equPremiseEpsilon})}\\
		&\geq \alpha\beta \frac{\sts(\mathcal{P}_2)}{\sts(G)} \stext{from Equation (\ref{equSpanningTreeScoreUpperAndLowerBounds})}\\
		&= \alpha \Pr_{\dist}[\mathcal{P}_2]. \qedhere
	\end{align*}
\end{proof}

As we mentioned in Section~\ref{subMainResult}, Theorem~\ref{thmSpanningTreeDistribution} is more easily understood through a corollary that shows an inverse exponential relationship between the ratio of probabilities (under the spanning tree distribution) and ratio of discrete perimeters of any two $m$-partitions. Here we state and prove a more formal version of the corollary.

\begin{corollary}
\label{corMain}
On any class of graphs that are $(k_1, k_2)$-bounded for constants $k_1,k_2\in\zz_{\geq 1}$, and for any pair of $m$-partitions $\mathcal{P}_1$ and $\mathcal{P}_2$,
    $$\frac{\Pr_{\dist}[\mathcal{P}_1]}{\Pr_{\dist}[\mathcal{P}_2]} \geq 2^{\Theta\left(\frac{\abs{\partial\mathcal{P}_2}}{\abs{\partial\mathcal{P}_1}}\right)}.$$
\end{corollary}

\begin{proof}
    Given any pair of $m$-partitions $\mathcal{P}_1$ and $\mathcal{P}_2$ of a $(k_1, k_2)$-bounded graph $G$ (for any $m$), let
    $$\alpha := \frac{1}{2k_2}\left(1 + \frac{1}{k_1 - 1}\right)^{\left(\frac{\abs{\partial \mathcal{P}_2}}{\abs{\partial \mathcal{P}_1}} - 1\right)}$$
    and $\varepsilon := 1$. Observe that we may equivalently write
    $$\alpha = \frac{\frac{1}{2k_2}}{\left(1 - \frac{1}{k_1}\right)^{\left(\frac{\abs{\partial \mathcal{P}_2}}{\abs{\partial \mathcal{P}_1}} - \varepsilon\right)}},$$
    so
    $$\log(\alpha) = \log\left(\frac{1}{2k_2}\right) - \left(\frac{\abs{\partial \mathcal{P}_2}}{\abs{\partial \mathcal{P}_1}} - \varepsilon\right)\log\left(1 - \frac{1}{k_1}\right).$$
    Plugging this into Equation (\ref{equLambdaDefinition}), we have
    $$\lambda(k_1, k_2, \alpha, \varepsilon) = \frac{\log\left(\frac{1}{2k_2}\right) - \log\left(\frac{1}{2k_2}\right) + \left(\frac{\abs{\partial \mathcal{P}_2}}{\abs{\partial \mathcal{P}_1}} - \varepsilon\right)\log\left(1 - \frac{1}{k_1}\right)}{\log\left(1 - \frac{1}{k_1}\right)} + \varepsilon = \frac{\abs{\partial \mathcal{P}_2}}{\abs{\partial \mathcal{P}_1}},$$
    so Equation (\ref{equPremiseBoundaryGap}) holds. Furthermore, Equation (\ref{equPremiseEpsilon}) must always hold for $\varepsilon = 1$, for otherwise $G$ would have to be disconnected. As long as $\frac{\abs{\partial \mathcal{P}_2}}{\abs{\partial \mathcal{P}_1}}$ is sufficiently large, we have $\alpha \geq 1$ as well, and thus we meet all of the hypotheses of Theorem \ref{thmSpanningTreeDistribution}, concluding that
    $$\frac{\Pr_{\dist}[\mathcal{P}_1]}{\Pr_{\dist}[\mathcal{P}_2]} \geq \alpha = \frac{1}{2k_2}\left(1 + \frac{1}{k_1 - 1}\right)^{\left(\frac{\abs{\partial \mathcal{P}_2}}{\abs{\partial \mathcal{P}_1}} - 1\right)}.$$
    This shows that
    \begin{equation*}
        \frac{\Pr_{\dist}[\mathcal{P}_1]}{\Pr_{\dist}[\mathcal{P}_2]} \geq 2^{\Theta\left(\frac{\abs{\partial\mathcal{P}_2}}{\abs{\partial\mathcal{P}_1}}\right)}.\qedhere
    \end{equation*}
\end{proof}

Finally, we remark that the constant $\lambda$ from Theorem \ref{thmSpanningTreeDistribution} must have some dependence on $k_1$ and $k_2$, and so there is not a more general statement that applies to, say, all planar graphs. The following two theorems demonstrate that assuming fixed $k_1$ and $k_2$ is necessary, even when we impose the additional realistic restriction that graphs do not have multiple edges between any pair of vertices.

\newpage

\begin{theorem}\label{thmCounterexample}
	There exists an infinite family of graphs $G_1, G_2, G_3, \dots$ such that:
	\begin{itemize}
		\item For any positive integer $n$, there exists $k_2$ such that $G_n$ is $(4, k_2)$-bounded and does not have multiple edges between any pair of vertices.
		\item There is a sequence of 2-partitions $\mathcal{P}_1^n$ and $\mathcal{P}_2^n$ of $G_n$ such that
		$$\lim_{n \to \infty} \frac{\abs{\partial \mathcal{P}_1^n}}{\abs{\partial \mathcal{P}_2^n}} = \lim_{n \to \infty} \frac{\Pr_{\dist}[\mathcal{P}_1^n]}{\Pr_{\dist}[\mathcal{P}_2^n]} = 0.$$
	\end{itemize}
\end{theorem}

\ipnc{1}{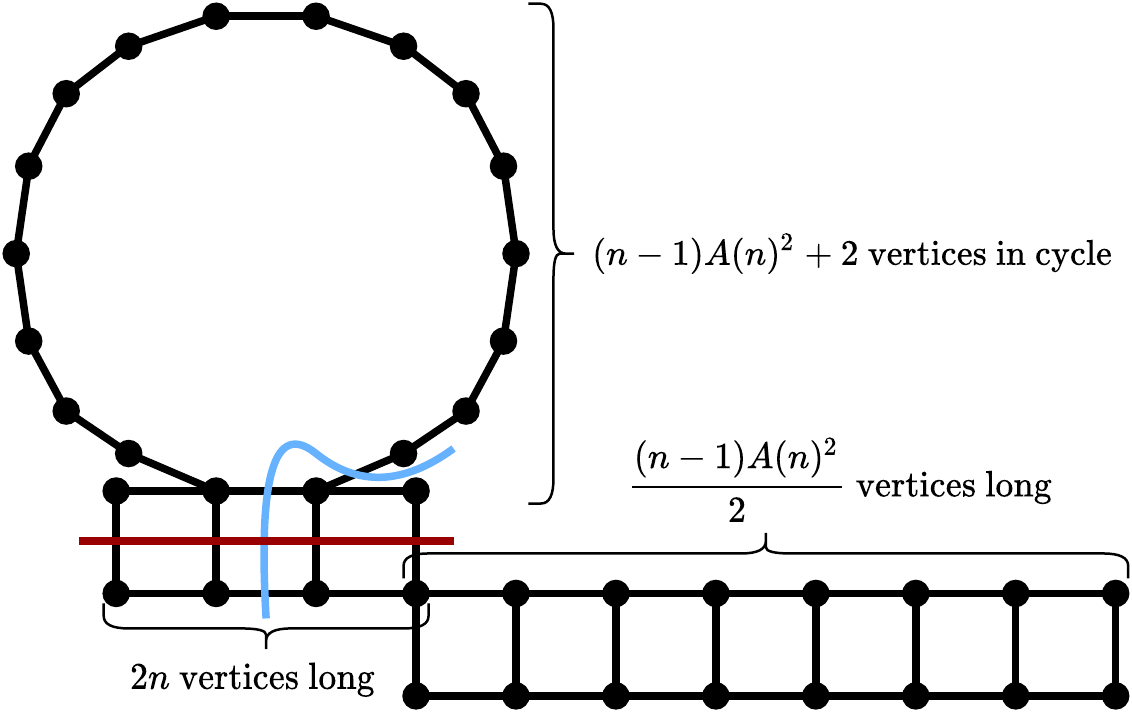}{\label{figCounterexample}A family of graphs $G_n$ (for even $n$) that are \emph{almost} $(4, 4)$-bounded (there are two bad faces $f_0$ instead of only one) for which Theorem \ref{thmSpanningTreeDistribution} fails. The instance shown is $G_2$.}

\begin{proof}
    Let $G_1, G_2, G_3, \dots$ be the family of graphs illustrated in Figure \ref{figCounterexample}, where
    $$A(n) := \threecases{\txt{if } n = 0}{0}{\txt{if } n = 1}{1}{\txt{for } n \geq 2}{4A(n - 1) - A(n - 2)}.$$
    Note that $A(n)$ is the number of spanning trees of a $2 \times n$ grid graph.\footnote{See: \url{http://oeis.org/A001353}} Let $\mathcal{P}_1^n$ be the 2-partition defined by the curved light blue line, and let $\mathcal{P}_2^n$ be the 2-partition defined by the horizontal dark red has line. Then
    $$\lim_{n \to \infty}\frac{\abs{\partial \mathcal{P}_1^n}}{\abs{\partial \mathcal{P}_2^n}} = \lim_{n \to \infty}\frac{3}{2n} = 0.$$
    Furthermore, it is easy to verify that $\mathcal{P}_1^n$ has a spanning tree score of $A(n)^2 A((n - 1)A(n)^2/2)$, whereas $\mathcal{P}_2^n$ has a spanning tree score of $((n - 1)A(n)^2 + 2) A((n - 1)A(n)^2/2)$. Therefore,
    $$\frac{\Pr_{\dist}[\mathcal{P}_1^n]}{\Pr_{\dist}[\mathcal{P}_2^n]} = \frac{\sts(\mathcal{P}_1^n)}{\sts(\mathcal{P}_2^n)} = \frac{A(n)^2}{(n - 1)A(n)^2 + 2} \leq \frac{1}{n - 1},$$
    which vanishes as $n \to \infty$.
\end{proof}

\begin{theorem}\label{thmDualCounterexample}
	There exists an infinite family of graphs $G_1, G_2, G_3, \dots$ such that:
	\begin{itemize}
		\item For any positive integer $n$, there exists $k_1$ such that $G_n$ is $(k_1, 7)$-bounded and does not have multiple edges between any pair of vertices.
		\item There is a sequence of 2-partitions $\mathcal{P}_1^n$ and $\mathcal{P}_2^n$ of $G_n$ such that
		$$\lim_{n \to \infty} \frac{\abs{\partial \mathcal{P}_1^n}}{\abs{\partial \mathcal{P}_2^n}} = \lim_{n \to \infty} \frac{\Pr_{\dist}[\mathcal{P}_1^n]}{\Pr_{\dist}[\mathcal{P}_2^n]} = 0.$$
	\end{itemize}
\end{theorem}

For Theorem \ref{thmCounterexample}, the main idea is to start with a large cycle, so a natural approach for Theorem \ref{thmDualCounterexample} would be to start with the dual of a large cycle; but this is a bundle of multiple edges between the same pair of vertices, which does not satisfy the conditions of the theorem. Due to this obstacle, the construction for Theorem \ref{thmDualCounterexample} is much more involved, so we defer it to the Appendix.

\section{Conclusion}\label{secConclusion}

Typically, for a heuristic sampling algorithm to be useful in practice it is not necessary to have theoretical guarantees; it merely has to ``just work.'' Unfortunately, this is clearly not the case for a problem so hotly contentious as political redistricting. One can imagine a plethora of creative ways to efficiently construct an ensemble of random graph partitions, but unless one is able to understand the distribution from which the samples are drawn, the properties of the ensemble may be meaningless from a statistical standpoint, and possibly from a legal standpoint as well.

For example, the 2017 gerrymandering court case \emph{League of Women Voters of Pennsylvania v.\txt{} Commonwealth of Pennsylvania} heard expert testimony from multiple mathematicians and political scientists using similar ensemble-based algorithms. In rebutting the use of an algorithm by Professor Jowei Chen \cite{ChenReport} to statistically conclude that the current map was gerrymandered, Professor Wendy K.\txt{} Tam Cho \cite{ChoReport} writes, ``Chen purports to have an algorithm that randomly generates maps. He has never evaluated this claim in any rigorous way. In my assessment of this `random' framework algorithm on a very small toy redistricting data set, I found that the strategy generated a biased set of maps that oversamples some maps while undersampling other maps.''

The ReCom algorithm is remarkable in that it simultaneously runs quickly and samples from a distribution that can be explicitly described. However, the description in terms of spanning trees still leaves much to be desired, and while we are not legal experts, we believe that trying to explain the concept to a court would be a nontrivial task. Our result provides the first known link between the spanning tree score and a more intuitive measure of compactness. We believe that understanding such relationships from a theoretical perspective is of great importance, especially given the fact that theorems have increasingly been playing a major role in the legal debate surrounding redistricting.

\section*{Acknowledgements}

The authors are deeply grateful to Moon Duchin and Daryl DeFord for helping us to understand the key open questions surrounding state-of-the-art redistricting algorithms. We are also grateful to our SODA reviewers for their careful reading and thoughtful comments, and, in particular, for a suggestion that led to Theorem \ref{thmDualCounterexample}.

This material is based upon work supported by the National Science Foundation Graduate Research Fellowship Program under grant DGE-1745303; by the National Science Foundation under grants CCF-2007080, IIS-2024287 and CCF-1733556; and by the Office of Naval Research under grant N00014-20-1-2488. Any opinions, findings, and conclusions or recommendations expressed in this material are those of the authors and do not necessarily reflect the views of the National Science Foundation or the Office of Naval Research.

\bibliographystyle{plain}
\bibliography{misc,ultimate}

\section*{Appendix}

\ipnc{1}{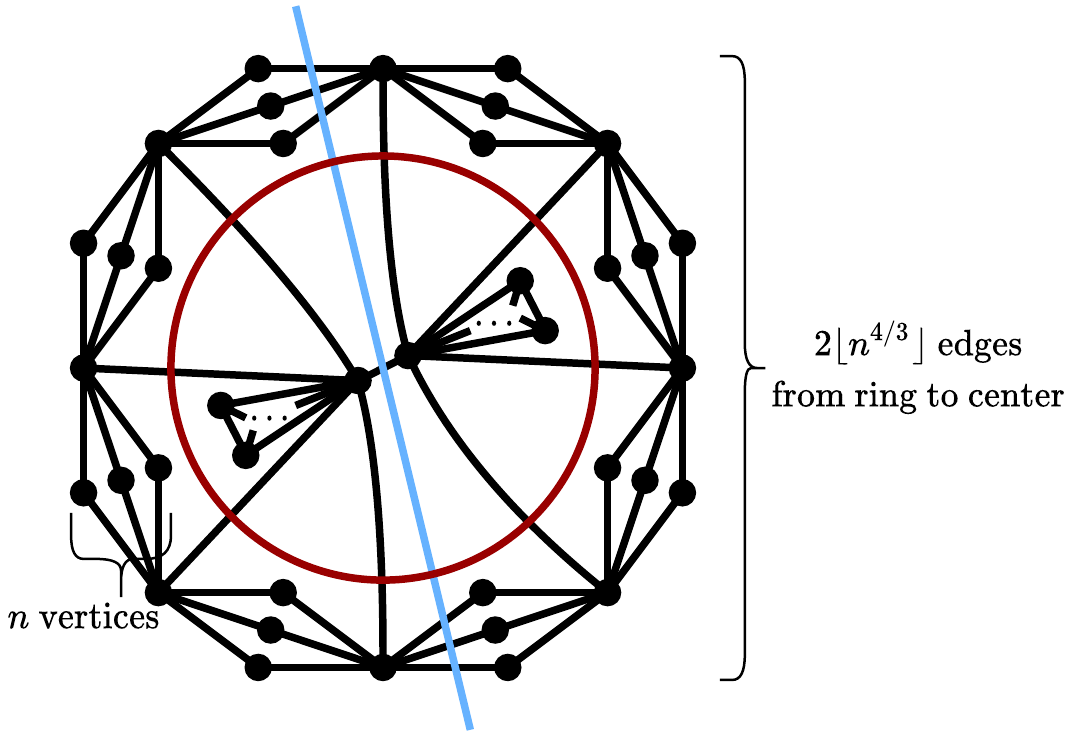}{\label{figDualCounterexample}A family of graphs $G_n$ that are each $(k_1, 7)$-bounded for increasing $k_1$, on which Theorem \ref{thmSpanningTreeDistribution} fails. The triangular subgraphs near the center can be arbitrary, and can be used to ensure the red partition is balanced. As $n \to \infty$, the red circular 2-partition will simultaneously have many more cut edges and a much greater probability of being sampled than the blue straight 2-partition. The particular instance shown is $G_3$.}

\begin{proof}[Proof of Theorem \ref{thmDualCounterexample}]
	Let $G_n$ be as depicted in Figure \ref{figDualCounterexample}, let $\mathcal{P}_1^n$ be the 2-partition defined by the blue straight line, and let $\mathcal{P}_2^n$ be the 2-partition defined by the red circle. It is possible to fill in the two triangular subgraphs so that both partitions are balanced and each $G_n$ is $(k_1, 7)$-bounded for some $k_1$, since all faces except the outer face have degree at most 7. Also,
	$$\lim_{n \to \infty} \frac{\abs{\partial \mathcal{P}_1^n}}{\abs{\partial \mathcal{P}_2^n}} = \lim_{n \to \infty} \frac{2n + 1}{2\lfloor n^{4/3} \rfloor} = 0.$$
	Thus, all that remains is to bound the probabilities of each partition being sampled, i.e., the ratio of their spanning tree scores. For this, we apply the same counting technique as in the proof Theorem \ref{thmSpanningTreeDistribution}.
	
	We begin with $\mathcal{P}_1^n$. Imagine a run of Algorithm \ref{algTreeSample} on $G_n$ in which we delete every edge in $\partial \mathcal{P}_1^n$ except for the central edge. Let the effective resistances computed by the algorithm be $\seq{r}{2n}$. Since each deleted edge is incident to a vertex of degree 2, by Lemma \ref{lemERDegree} we know that $r_i \geq \frac12$ for all $i$. Thus,
	\begin{align*}
		\frac{\sts(\mathcal{P}_1^n)}{\sts(G_n)} &= \prod_{i = 1}^{2n} (1 - r_i)\\
		&\leq \prod_{i = 1}^{2n} \left(1 - \frac12\right)\\
		&= \frac{1}{4^n}
	\end{align*}
	
	For $\mathcal{P}_2^n$, we imagine a run of Algorithm \ref{algTreeSample} on $G_n$ in which we delete edges intersecting the red circle in clockwise order, starting from where the red circle intersects the blue line (e.g., 12:00 in Figure \ref{figDualCounterexample}). Label the edges in the order of deletion,
	$$e_{\lfloor n^{4/3} \rfloor - 1}, e_{\lfloor n^{4/3} \rfloor - 2}, e_{\lfloor n^{4/3} \rfloor - 3}, \dots, e_2, e_1, e_0, e'_{\lfloor n^{4/3} \rfloor - 1}, e'_{\lfloor n^{4/3} \rfloor - 2}, e'_{\lfloor n^{4/3} \rfloor - 3}, \dots, e'_2, e'_1,$$
	and let $e'_0$ be the final edge in the circle, which is not deleted. Let $r_i$ and $r'_i$ denote the respective effective resistances of $e_i$ and $e'_i$ in the graph obtained by removing all previously deleted edges.
	
	Since $e_0$ is contained within a cycle of length 5 when deleted, we know by Lemma \ref{lemERCycle} that $r_0 \leq 1 - \frac15$, so $(1 - r_0) \geq \frac15$. All that remains is to compute upper bounds on $r_i$ and $r_i'$ for $i \geq 1$. For both cases, we apply Lemma \ref{lemRayleigh}, deleting all edges except for the subgraph containing half of the outer ring, i.e., the subgraph on the left of Figure \ref{figDualCounterexampleProof}. By symmetry, the analysis for bounding $r'_i$ is the same as for $r_i$, so we only consider $r_i$.
	
	\ipnc{1}{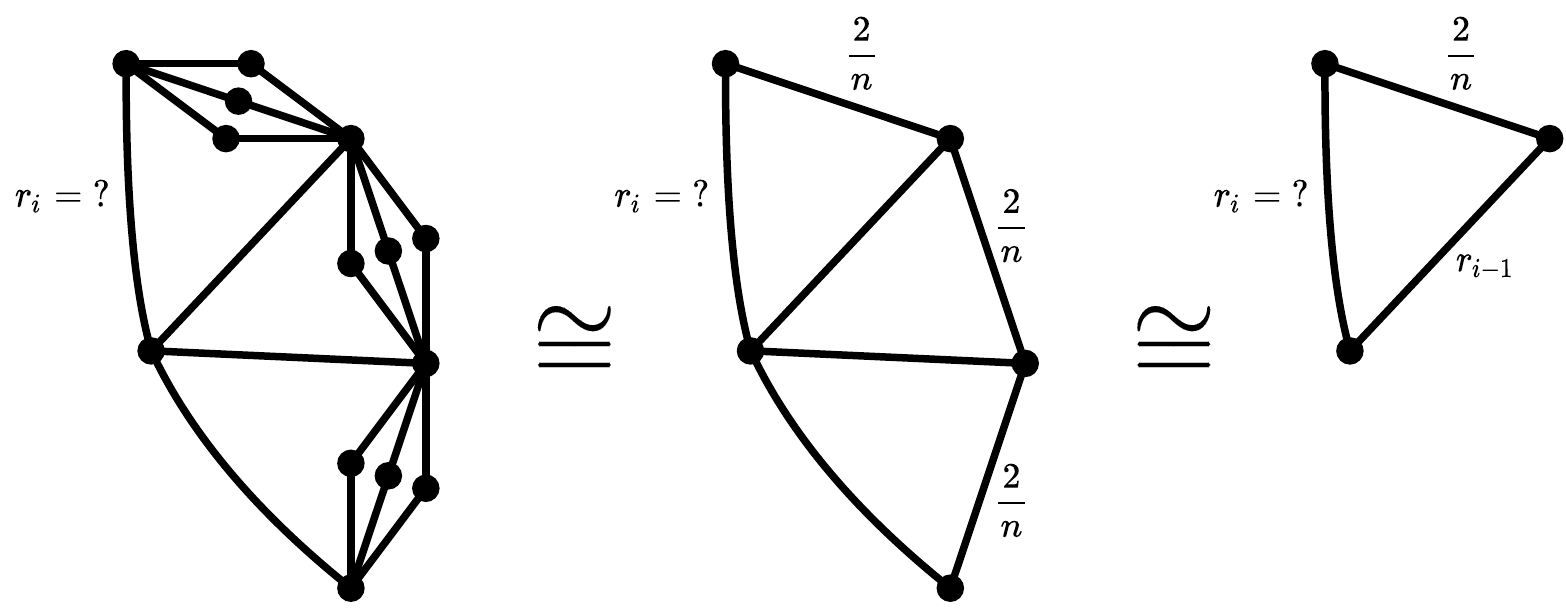}{\label{figDualCounterexampleProof}We compute the effective resistance $r_i$ by replacing subgraphs with single edges whose effective resistances in those subgraphs are already known.}
	
	We claim that, for all $i$,
	$$r_i \leq \frac{2}{\min\left\{i, \left\lfloor\frac{\sqrt{n}}{2}\right\rfloor\right\} + 2}.$$
	We proceed by induction on $i$ (in order of increasing $i$, which is the reverse of the order in which we actually delete the edges). For the base case, $i = 0$, this states that $r_i \leq 1$, which is always true. There are two inductive cases to consider.
	
	First suppose the claim holds for $i - 1$, where $1 \leq i \leq \left\lfloor\frac{\sqrt{n}}{2}\right\rfloor$. Then, using standard series/parallel laws, we can replace subgraphs of unit resistors by non-unit resistors according to their effective resistances, as illustrated in Figure \ref{figDualCounterexampleProof}. In the final graph, we then have
	\begin{align*}
		r_i &= \frac{1}{1 + \frac{1}{r_{i - 1} + \frac2n}}\\
		&\leq \frac{1}{1 + \frac{1}{\frac{2}{i + 1} + \frac2n}} \indhyp\\
		&= \frac{1}{1 + \frac{n(i + 1)}{2n + 2i + 2}}\\
		&= \frac{2n + 2i + 2}{ni + n + 2n + 2i + 2}\\
		&\leq \frac{2}{i + 2},
	\end{align*}
	where the final equality follows from cross-multiplying:
	\begin{align*}
		i \leq \frac{\sqrt{n}}{2} &\implies i^2 + i \leq n\\
		&\implies 2i^2 + 2i \leq 2n\\
		&\implies 2ni + 2i^2 + 2i + 4n + 4i + 4 \leq 2ni + 6n + 4i + 4\\
		&\implies (2n + 2i + 2)(i + 2) \leq 2(ni + 3n + 2i + 2).
	\end{align*}
	
	Now instead suppose $i > \left\lfloor\frac{\sqrt{n}}{2}\right\rfloor$. In this case, observe that, when we compute the effective resistance in the final graph in Figure \ref{figDualCounterexampleProof},
	\begin{align*}
		r_i &= \frac{1}{1 + \frac{1}{r_{i - 1} + \frac2n}}\\
		&\leq \frac{1}{1 + \frac{1}{\frac{2}{\frac{\sqrt{n}}{2} + 2} + \frac2n}} \indhyp\\
		&= \frac{1}{1 + \frac{1}{\frac{2 + \frac{1}{\sqrt{n}} + \frac4n}{\frac{\sqrt{n}}{2} + 2}}}\\
		&= \frac{1}{1 + \frac{\frac{\sqrt{n}}{2} + 2}{2 + \frac{1}{\sqrt{n}} + \frac4n}}\\
		&= \frac{2 + \frac{1}{\sqrt{n}} + \frac{4}{n}}{\frac{\sqrt{n}}{2} + 4 + \frac{1}{\sqrt{n}} + \frac{4}{n}}\\
		&\leq \frac{2}{\frac{\sqrt{n}}{2} + 2},
	\end{align*}
	where the final equality follows from cross-multiplying:
	\begin{align*}
		\frac{1}{\sqrt{n}} \leq 1 &\implies \frac{2}{\sqrt{n}} \leq 3.5\\
		&\implies \sqrt{n} + \frac12 + \frac{2}{\sqrt{n}} + 4 + \frac{2}{\sqrt{n}} + \frac{8}{n} \leq \sqrt{n} + 8 + \frac{2}{\sqrt{n}} + \frac{8}{n}\\
		&\implies \left(2 + \frac{1}{\sqrt{n}} + \frac4n\right)\left(\frac{\sqrt{n}}{2} + 2\right) \leq 2\left(\frac{\sqrt{n}}{2} + 4 + \frac{1}{\sqrt{n}} + \frac{4}{n}\right).
	\end{align*}
	By induction, the claim holds for all $i$.
	
	Note that, as long as $n$ is sufficiently large, for $1 \leq i \leq \left\lfloor\frac{\sqrt{n}}{2}\right\rfloor$ this implies that $r_i, r'_i \leq \frac23$, and for $\left\lfloor\frac{\sqrt{n}}{2}\right\rfloor + 1 \leq i \leq \lfloor n^{4/3} \rfloor - 1$ this implies that
	$$r_i, r'_i \leq \frac{2}{\left\lfloor\frac{\sqrt{n}}{2}\right\rfloor + 2} \leq \frac{4}{\sqrt{n}}.$$
	Putting these bounds all together, we have
	\begin{align*}
		\frac{\sts(\mathcal{P}_2^n)}{\sts(G_n)} &= \left(\prod_{i = \lfloor \sqrt{n}/2 \rfloor + 1}^{\lfloor n^{4/3} \rfloor - 1} (1 - r_i)\right) \left(\prod_{i = 1}^{\lfloor \sqrt{n}/2 \rfloor} (1 - r_i)\right) \push\cdot (1 - r_0) \left(\prod_{i = \lfloor \sqrt{n}/2 \rfloor + 1}^{\lfloor n^{4/3} \rfloor - 1} (1 - r'_i)\right) \left(\prod_{i = 1}^{\lfloor \sqrt{n}/2 \rfloor} (1 - r'_i)\right)\\
		&\geq  \left(\prod_{i = \lfloor \sqrt{n}/2 \rfloor + 1}^{\lfloor n^{4/3} \rfloor - 1} \left(1 - \frac{4}{\sqrt{n}}\right)\right) \left(\prod_{i = 1}^{\lfloor \sqrt{n}/2 \rfloor} \left(1 - \frac23\right)\right) \push\cdot\left(\frac15\right) \left(\prod_{i = \lfloor \sqrt{n}/2 \rfloor + 1}^{\lfloor n^{4/3} \rfloor - 1} \left(1 - \frac{4}{\sqrt{n}}\right)\right) \left(\prod_{i = 1}^{\lfloor \sqrt{n}/2 \rfloor} \left(1 - \frac23\right)\right)\\
		&\geq \frac15 \left(\left(1 - \frac{4}{\sqrt{n}}\right)^{n^{4/3}} \left(\frac{1}{3}\right)^{n^{1/2}/2}\right)^2\\
		&= \frac15 \left(\left(\left(1 - \frac{4}{\sqrt{n}}\right)^{\frac{\sqrt{n}}{4}}\right)^{\frac{n^{4/3}}{\frac{\sqrt{n}}{4}}} \left(\frac{1}{3}\right)^{n^{1/2}/2}\right)^2\\
		&= \frac15 \left(\left(\left(1 - \frac{4}{\sqrt{n}}\right)^{\frac{\sqrt{n}}{4}}\right)^{4n^{5/6}} \left(\frac{1}{3}\right)^{n^{1/2}/2}\right)^2\\
		&\geq \frac15 \left(\left(\frac{1 - \frac{4}{\sqrt{n}}}{e}\right)^{4n^{5/6}} \left(\frac{1}{3}\right)^{n^{1/2}/2}\right)^2\\
		&\geq \frac15 \left(\left(\frac{1}{3}\right)^{4n^{5/6}} \left(\frac{1}{3}\right)^{n^{1/2}/2}\right)^2 \stext{for large enough $n$}\\
		&= \frac15 \cdot \frac{1}{3^{8n^{5/6} + n^{1/2}}}\\
		&\geq \frac15 \cdot \frac{1}{4^{9n^{5/6}}}.
	\end{align*}
	
	It follows that
	$$\frac{\Pr_{\dist}[\mathcal{P}_1^n]}{\Pr_{\dist}[\mathcal{P}_2^n]} = \frac{\sts(\mathcal{P}_1^n)/\sts(G_n)}{\sts(\mathcal{P}_2^n)/\sts(G_n)} \leq \frac{\frac{1}{4^n}}{\frac15 \cdot \frac{1}{4^{9n^{5/6}}}} = \frac{5 \cdot 4^{9n^{5/6}}}{4^n} = 5 \cdot 4^{9n^{5/6} - n},$$
	which vanishes as $n \to \infty$.
\end{proof}

\end{document}